\documentclass[runningheads]{llncs}

\usepackage[T1]{fontenc}

\usepackage[margin=1.5in]{geometry}
\usepackage{soul}
\usepackage{url}
\usepackage[hidelinks]{hyperref}
\usepackage[utf8]{inputenc}
\usepackage[small]{caption}
\usepackage{graphicx}
\usepackage{amsmath}
\usepackage{booktabs}
\usepackage{algorithm}
\usepackage{algorithmic}
\usepackage[switch]{lineno}

\usepackage[square,sort,numbers,sectionbib]{natbib} 

 \makeatletter 
 \renewcommand\@biblabel[1]{#1} 
 \makeatother

\usepackage{amsfonts,amssymb,bbm,mathtools}
\usepackage{comment,cprotect,xcolor,xspace}
\usepackage{multirow}
\usepackage{cleveref} 

\newcommand{\bahar}[1]{\textcolor{orange}{\textbf{Bahar says:} #1}}

\newcommand{\haris}[1]{\textcolor{red}{\textbf{Haris says:} #1}}

\newcommand{\profileset}{\mathbf{{A}}} 
\newcommand{\profile}{\mathcal{A}} 

\newcommand{\plausible}{\emph{plausible}\xspace}

\newcommand{\jrProb}{\textsc{JR-Probability}\xspace}

\newcommand{\isPosJR}{\textsc{IsPossJR}\xspace}

\newcommand{\isNecJR}{\textsc{IsNecJR}\xspace}

\newcommand{\existsPosJR}{\textsc{ExistsPossJR}\xspace}

\newcommand{\existsNecJR}{\textsc{ExistsNecJR}\xspace}

\begin{document}

\title{Approval-Based Committee Voting under Uncertainty}

\author{
Hariz Aziz\inst{1} \and
Venkateswara Rao Kagita\inst{2} \and
Baharak Rastegari \inst{3} \and
Mashbat Suzuki\inst{1}
}
\authorrunning{H. Aziz et al.}
\institute{
University of New South Wales, Sydney, Australia \\ \email{\{haris.aziz,mashbat.suzuki\}@unsw.edu.au} \and
National Institute of Technology, Warangal, India \\ \email{venkat.kagita@nitw.ac.in} \and
University of Southampton, Southampton, UK \\ \email{b.rastegari@soton.ac.uk}
}

\maketitle

\begin{abstract}
We study approval-based committee voting in which a target number of candidates are selected based on voters' approval preferences over candidates. In contrast to most of the work, we consider the setting where voters express uncertain approval preferences and explore four different types of uncertain approval preference models. For each model, we study the problems such as computing a committee with the highest probability of satisfying axioms such as justified representation.
\keywords{Approval preferences, Committee voting, ABC voting }
\end{abstract}

\section{Introduction}


Multi-agent collective decision making is one of the fundamental issues in computer science and economics.
We consider the ubiquitous problem of selecting a given number of candidates based on voters' approval preferences over candidates. The problem has been referred to as \textit{approval-based committee (ABC) voting } or multi-winner voting under approval preferences~\citep{Kilg10a,LaSk23a}. 
Prominent concerns in ABC voting include proportional representation of voters and finding committees that are desirable and representative. One of the most widely studied concepts for fairness in ABC voting is\textit{ justified representation (JR)} and stronger variants~\citep{AEH+18a,BFJL16a, LaSk23a}. 

Almost all of the focus in ABC voting has been on voters with deterministic preferences. While proportional representation in ABC voting has received tremendous interest, especially in recent years (see, e.g., ~\citep{LaSk23a}),  there is little prior work on computing desirable representative committees under \textit{uncertain} preferences.  Uncertain approval preferences are useful when the central planner only has imprecise information about the voters' preferences. This estimated information could be based on historical preferences, past selections, or online clicks or views. For example, if a voter $i$ has selected a certain candidate $c$ 70\% of the times in previous situations, one could use that information to assume that the approval probability of voter $i$ for candidate $c$ is 0.7. The information could also be based on situations where each voter actually represents a group of people who may not have identical approval preferences. For example, if 60\% of the group approved a certain candidate, one could assume that approval probability of voter $i$ for candidate $c$ is 0.6. 
The estimated information could also represent the confidence of a machine learning or a recommendation technique that is employed to predict the unobserved (dis)approvals of voters for candidates.
Collaborative filtering techniques can be applied in scenarios with a partially observed matrix depicting voter-candidate interactions through approval/disapproval entries~\cite{himabindu2018conformal,kagita2017conformal}. These techniques enable the prediction of unknown entries with a confidence level of $(1-\varepsilon)$, where $\varepsilon$ represents the probability of making an error.
To illustrate, consider an algorithm that predicts approval with a confidence level of 70\%. In this case, the probability that a voter approves a candidate is 0.7, while the probability of disapproval is 0.3. 

In this paper, we initiate work on problems where voters' uncertain approval preferences are taken into account to compute desirable committees that satisfy representation with high probability.
We consider four different types of uncertain approval preferences.
In the \emph{Joint Probability model}, there is a probability distribution of approval profiles.
In the \emph{Lottery model}, each voter has an independent probability distribution over approval sets.
In the \emph{Candidate-Probability model}, there is a probability for a given voter approving a given candidate.
We also consider a restricted version of the 
latter
model.
For each of the uncertain approval models, we consider problems such as computing a committee with the highest probability of being JR.
Such an outcome can be viewed as being  fair under uncertain information.

\paragraph{Contributions}

We undertake a detailed computational complexity analysis of several problems with respect to the four preference uncertainty models. 
For each of the preference models, we consider problems such as  {\sc JR-Probability} (computing the probability that a given committee satisfies JR); {\sc IsPossJR} (deciding whether a given committee satisfies JR  with non-zero probability); {\sc IsNecJR} (deciding whether a given committee satisfies JR with probability one); {\sc ExistsNecJR} (deciding whether there exists a committee that satisfies JR with probability one); and {\sc MaxJR} (computing a committee  that has the highest probability of satisfying JR). Our central focus is on justified representation (JR) but we also present several results for stronger proportional representation axioms such as PJR and EJR. 
We show that {\sc ExistsNecJR} is NP-complete for all the models that we consider. {\sc IsPossJR} is computationally hard for the lottery model and polynomial-time solvable for the other models.  Our results are summarized in Table~\ref{table:summary:uncertainABC}. 

         \begin{table*}[t]
\setlength{\tabcolsep}{-1pt}
             \footnotesize
             \centering
             \scalebox{0.75}{
             \begin{tabular}{p{2.5cm} @{\qquad} p{2.5cm} @{\qquad} p{2.8cm} @{\qquad} p{3cm} @{\qquad} p{3cm}}
                 \toprule
              &\textbf{Joint Prob.}&\textbf{Lottery}&\textbf{Candidate Prob.}&\textbf{3VA} \\
    	   \textbf{Problems}&&\\
                  \midrule\\
              {\sc JR-Probability} &in P & NP-h (C.~\ref{cor:JR-prob-lottery}) &  \#P-complete (T. \ref{th:counting}) & \#P-complete (T. \ref{th:counting}), in P  if $p_{i,c}\in \{0,1\}$, $\forall i \in V$ \& $\forall c\in W$ (T.~\ref{thm:JR-prob-3VA-Pcase1}) \\
                  \midrule 
              {\sc IsPossJR} &in P (remark) & NP-c (T.~\ref{thm:IsPossJr-Lottery}) &in P (Mod. of \citep{IBK22a}) &in   P \citep{IBK22a} \\
                  \midrule 
              \begin{tabular}{l} {\hspace{-0.3cm} \sc IsPossEJR/}\\\hspace{-0.2cm}{\sc IsPossPJR}\end{tabular} & coNP-c (C.~\ref{cor:IsNecEJR-IsPossEJR}) & coNP-c (C.~\ref{cor:IsNecEJR-IsPossEJR}) &   coNP-c (C.~\ref{cor:IsNecEJR-IsPossEJR}) & coNP-c (C.~\ref{cor:IsNecEJR-IsPossEJR}) \\

                  \midrule

              {\sc IsNecJR} & in P (remark) & in P (T.~\ref{thm:IsNecJR-Lottery}) &in P (Mod. of \citep{IBK22a}) &in  P \citep{IBK22a}  \\
        \midrule
              \begin{tabular}{l} {\hspace{-0.3cm} \sc IsNecEJR/}\\\hspace{-0.2cm}{\sc IsNecPJR}\end{tabular} & coNP-c (C.~\ref{cor:IsNecEJR-IsPossEJR}) & coNP-c (C.~\ref{cor:IsNecEJR-IsPossEJR}) &   coNP-c (C.~\ref{cor:IsNecEJR-IsPossEJR}) & coNP-c (C.~\ref{cor:IsNecEJR-IsPossEJR}) \\
                  \midrule
              {\sc ExistsNecJR} & NP-c (T.~\ref{thm:ExistsNecJRLotteryJoint}), even if there are only 2 \plausible approval profiles (T.~\ref{thm:ExistsNecJRjoint2}) & NP-c (T.~\ref{thm:ExistsNecJRLotteryJoint}), in P if all \plausible approval sets are of size 1 (T.~\ref{thm:IsNecJR-Lottery-P}) & NP-c (T.~\ref{thm:ExistsNecJRCandidate}), in P if $0<p_{i,c}<1$ $\forall i\in V \& \forall c\in C$ (Cor. of L.~\ref{lem:fullset})& NP-c (C.~\ref{thm:ExistsNecJR3VA}), in P if $0<p_{i,c}<1$ $\forall i\in V \& \forall c\in C$ (Cor. of L.~\ref{lem:fullset}) \\ 
 \midrule
              \begin{tabular}{l} {\hspace{-0.3cm} \sc ExistsNecEJR/}\\\hspace{-0.2cm}{\sc ExistsNecPJR}\end{tabular} & coNP-h (T.~\ref{thm:ExistsNecEJR-Lottery-JP}) & coNP-h (T.~\ref{thm:ExistsNecEJR-Lottery-JP}) &   coNP-h  (T.~\ref{thm:ExistsNecEJR-3VA-CP})  & coNP-h  (T.~\ref{thm:ExistsNecEJR-3VA-CP})  \\ 
                  \midrule
\begin{tabular}{l} {\hspace{-0.3cm} \sc MaxJR}\end{tabular} & NP-h (C.~\ref{cor:MaxJR})  &  NP-h (C.~\ref{cor:MaxJR})  &  NP-h(C.~\ref{cor:MaxJR})   & NP-h (C.~\ref{cor:MaxJR}) \\
                  \midrule
\begin{tabular}{l} {\hspace{-0.3cm} \sc MaxEJR/}\\ \hspace{-0.2cm}{\sc MaxPJR}\end{tabular} & coNP-h (C.~\ref {cor:MaxEJR-MaxPJR})  & coNP-h (C.~\ref {cor:MaxEJR-MaxPJR}) & coNP-h (C.~\ref {cor:MaxEJR-MaxPJR})  & coNP-h (C.~\ref {cor:MaxEJR-MaxPJR}) \\
                 \bottomrule
             \end{tabular}
        }
         \caption{Summary of results.  Please see Sections \ref{sec:preliminaries} for notations and definitions. }

             \label{table:summary:uncertainABC}
         \end{table*}
\section{Related Work}

Our central model is  approval-based committee (ABC) voting. One of the main questions within ABC voting  is how to produce committees which choose candidates ``proportional'' to the support they receive from voters.
\citet{ABC+16a} initiated a study of approval-based committee voting based on the idea of ``justified representation'' for cohesive groups.
The study has led to a  large body of work focusing on representation axioms and voting rules which produce committees satisfying these axioms \citep{AEH+18a,BFJL16a} look more at query complexity and approximate (E)JR
For a detailed survey of the recent work on approval-based committee voting, we refer the readers to the book of \citet{LaSk23a}. \citet{BFKN19a} take an experimental view of committees satisfying JR/PJR/EJR and also examine the complexity of counting committees satisfying these properties.

Uncertainty in preferences in voting problems has received interest in many papers.
\citet{KoLa05a} study winner determination with incomplete preferences under ranking- based single-winner rules (such as plurality and Borda). \citet{HAK+12a}  examine the probability of a particular candidate winning an election under uncertain preferences for various voting rules such as Plurality and Borda. {\citet{MeLa19a} studied facility location under uncertain information.}
Our work combines aspects of uncertainty in preferences with approval-based committee voting.

The two papers, recently published, that are most closely related to ours are \cite{IBK22a} and \cite{HKPTW23}.
In one of the first detailed papers to consider approval-based committee voting under uncertain preferences,  \citet{IBK22a} proposed several models of incomplete information. Unlike our work, the input information in their paper does not have quantitative information about probabilities and the authors do not examine problems such as maximizing the probability of satisfying an axiomatic property such as justified representation. Three of the four uncertain preference models that we consider are not explored by \citet{IBK22a}. One of the models that we consider (3VA) was also studied by \citet{IBK22a}.  For their incomplete information models, 
they consider the problem of checking whether a given committee is possibly or necessarily JR. For our uncertain preferences models, we also consider other problems such as checking the existence of possibly or necessarily JR and also computing the committee that maximizes the probability of satisfying JR. \citet{HKPTW23} looked at query complexity and approximating JR and EJR for the problem of selecting comments based on agreements and disagreements in platforms for online civic participation. Their set-up is equivalent to ABC voting. In their work voters (users of the platform) are randomly chosen and queried to express their approval/disapproval on a small set of candidates (comments), generating uncertaint/noisy information. Previously, \citet{BLMR13}  also examined single and multi-winner voting  under uncertain approvals.  \citet{ZKO21a} examined the problem of  checking whether an incomplete approval profile admits a completion within a certain restricted domain of approval preferences.

Uncertain preferences have been explored in other contexts as well, including matching and allocations.  \citet{ABG+20a} study the problem of computing stable matchings under uncertain preferences in the context of two-sided matching. 
 \citet{ABH+19} undertake a detailed complexity analysis of Pareto optimal allocation under uncertain preferences. \citet{RSS24} examine the envy-free allocation under uncertain preferences.  
The Lottery and Joint Probability models that we study are inspired by similar models in the context of allocation and matching problems mentioned above.
The Candidate-Probability model that we introduce in the paper captures scenarios where the approval probability may be based on a history of past choices.

\section{Preliminaries}
\label{sec:preliminaries}

For any~$t \in \mathbb{N}$, let $[t] \coloneqq \{1, 2, \dots, t\}$.
An \emph{instance} of the (deterministic) approval-based committee (ABC) voting is represented as a tuple $(V, C, \profile, k)$, where:
\begin{itemize}
\item $V = [n]$ and $C = [m]$ are the set of \emph{voters} and \emph{candidates}, respectively.

\item An approval set is a subset of candidates.
Denote by~$A_i$ voter~$i$'s approval set.
$\profile = (A_1, A_2, \dots, A_n)$ is voters' \emph{approval profile}.
The set of all possible approval profiles is denoted as~$\profileset$.

\item $k$ is a positive integer representing the committee size.
\end{itemize}
A \emph{winning committee}~$W \subseteq C$ is of size~$k$. In the following, we introduce proportional representation axioms, and begin with the central notion of this paper, which is \emph{justified representation (JR)}~\citep{ABC+16a}.

\begin{definition}[JR]
\label{def:JR}
Given instance $(V, C, \profile, k)$, a committee~$W$ is said to satisfy \emph{justified representation (JR)} if for every group of voters~$V' \subseteq V$ with $|V'| \geq \frac{n}{k}$ and $\bigcap_{i \in V'} A_i \neq \emptyset$, it holds that $A_i \cap W \neq \emptyset$ for some~$i \in V'$.
\end{definition}

JR has been strengthened to stronger notions, of which \emph{proportional justified representation (PJR)}~\citep{SFF+17a} and \emph{extended justified representation (EJR)}~\citep{ABC+16a} are two important ones.
In order to reason about the two notions, we first introduce the concept of a cohesive group.
For any positive integer~$\ell$, a group of voters~$V' \subseteq V$ is said to be \emph{$\ell$-cohesive} if $|V'| \geq \ell \cdot \frac{n}{k}$ and $|\bigcap_{i \in V'} A_i| \geq \ell$.

\begin{definition}[PJR and EJR]
Given instance $(V, C, \profile, k)$, a committee~$W$ is said to satisfy \emph{PJR} (resp., \emph{EJR}) if for every positive integer~$\ell$ and every $\ell$-cohesive group of voters~$V' \subseteq V$, it holds that $|\left( \bigcup_{i \in V'} A_i \right) \cap W| \geq \ell$ (resp., $|A_i \cap W| \geq \ell$ for some~$i \in V'$).
\end{definition}

As the definitions suggest, EJR $\implies$ PJR $\implies$ JR.

\subsection{Uncertain Preference Models}
\label{ssec:uncertain-preference-models}

We consider ABC voting where voters may be \emph{uncertain} about their approval ballots.
Specifically, we consider the following uncertainty models:
\begin{enumerate}
\item \textbf{Joint Probability model:}
We are given a probability distribution~$\Delta(\profileset) \coloneqq \{(\lambda_r, \profile_r)\}_{r \in [s]}$ over~$s$ approval profiles with $\sum_{r \in [s]} \lambda_r = 1$, where for each~$r \in [s]$, the approval profile~$\profile_r$ is associated with a positive probability~$\lambda_r>0$. {We write $\Delta(\profile_r) = \lambda_r$.}

\item \textbf{Lottery model:} For each voter $i\in V$, we are given a probability distribution~$\Delta_i(2^C) \coloneqq \{(\lambda_r, S_r)\}_{r \in [s_i]}$ over~$s_i$ approval sets with $\sum_{r \in [s_i]} \lambda_r = 1$, where for each~$r \in [s_i]$, the candidate set~$S_r\subseteq C$ is associated with a positive probability~$\lambda_r>0$. {We write $\Delta_i(S_r) = \lambda_r$.} 
We assume that the probability distributions of voters are independent.

\item \textbf{Candidate-Probability model:}
Each voter $i$ approves each candidate~$c$ 
independently with probability~$p_{i, c}$, i.e., for each~$i \in V$ and each~$c \in C$, $p_{i, c} \in [0, 1]$. 
\item \textbf{Three-Valued Approval (3VA) model:}
Each voter specifies a subset of candidates that are approved and a subset of candidates that are disapproved. The remaining candidates could be approved or disapproved with equal probability. 
That is, $p_{i,c}\in\{0, 0.5, 1\}, \forall i\in V, c\in C$, wherein 0 denotes disapproval, 1 indicates approval, and 0.5 represents unknown.
\end{enumerate}

The Joint Probability and Lottery models have been studied in other contexts including two-sided stable matching problem and assignment problem~\citep{ABH+19,ABG+20a}.
The 3VA model has been studied in the ABC voting~\citep{IBK22a}.
The Candidate-Probability model is a direct generalisation of the 3VA model. 
Under Candidate-Probability and 3VA models, $p_{i, c}$'s are assumed to be independent.
We refer to an approval profile that can occur with positive probability, under any of the uncertainty models, as a \textbf{\plausible} approval profile.

\begin{proposition}
There is a unique Joint Probability model representation for preferences given in the Lottery model.
\end{proposition}

\begin{proof}
The Joint Probability representation entails the probability distribution across each distinct profile of approval sets for all voters. In the Lottery model, the probability distribution of voters for various approval sets is known. The Joint Probability of a profile comprising approval sets from all voters can be obtained by multiplying the individual voter's approval probabilities i.e., $\Delta(\profile) = \prod_i \Delta_i(\profile_i), \forall \profile\in \profileset$.
\end{proof}

\begin{proposition}
There is a unique Lottery model representation for preferences given in Candidate-Probability model.
\end{proposition}

\begin{proof}
    The probability distribution across approval sets can be deduced from a probability distribution over candidates 
    albeit with an exponential increase in complexity. The probability of voter $i$ approving a candidate set $S$ can be computed as $\Delta_i(S) = \prod_{c\in S}(p_{i,c})\times \prod_{c\in C\setminus S} (1 - p_{i,c})$.
\end{proof}


\emph{Example:} This example illustrates the transition from an instance in the Candidate Probability model, where voters have probabilities for individual candidates, to an instance in the Lottery model, where probabilities are assigned to sets of candidates. Assume there are  3 candidates and one voter. The approval probabilities of voter $1$ for the three candidates are $p_{1, 1} = 0.9,~ p_{1, 2} = 0.6,$ and $p_{1, 3} = 0.5$. The probability of voter $1$ approving a candidate set, say $\{1, 3\}$, is calculated as $\Delta_1({1, 3}) = 0.9 \times (1-0.6) \times 0.5 = 0.18$.

\subsection{Computational Problems}

We are interested in extending the quest for a committee that satisfies JR, PJR or EJR to the realm of uncertainty.
For ease of exposition, we use JR as an exemplary property to introduce the most natural computational problems of interest.

\paragraph{Possible JR}
In \Cref{sec:possible-JR}, we examine computational problems related to \emph{possible JR}. 
Possible JR refers to the property that a committee has non-zero probability of satisfying JR.
Fix any uncertain preference model defined in \Cref{ssec:uncertain-preference-models}.
Given as input voters' uncertain preferences,
\begin{itemize}
\item \isPosJR is the problem of deciding whether a given committee~$W$ satisfies JR with non-zero probability;

\item \existsPosJR is the problem of deciding whether there exists a committee~$W$ that satisfies JR with non-zero probability.
\end{itemize}

\paragraph{Necessary JR}
In \Cref{sec:necessary-JR}, we examine computational problems related to \emph{necessary JR}.
Fix any uncertain preference model.
Given as input voters' uncertain preferences,
\begin{itemize}
\item \isNecJR is the problem of deciding whether a given committee~$W$ satisfies JR with probability~$1$;

\item \existsNecJR is the problem of deciding whether there exists a committee~$W$ that satisfies JR with probability~$1$.
\end{itemize}

\paragraph{JR Probability}
In \Cref{sec:JR-probability}, we examine the problem referred to as \jrProb. It is the  problem of computing the probability that a given committee~$W$ satisfies JR.
Since we are interested in computing desirable committees under uncertainty, we also provide results for the following problem: 
{\sc MaxJR}: 
Computing a committee $W$ that has the highest probability of satisfying JR. 
Note that  \existsNecJR reduces to {\sc MaxJR}. 

Finally, in \Cref{sec:PJR-and-EJR} we examine the same set of computational problems for PJR and EJR.

\section{Possible JR}
\label{sec:possible-JR}

In this section, we examine problems related to the Possible JR. Firstly, we note that {\sc ExistsPossJR}
trivially has a yes answer
for all of our probability models. This is due to the fact that for any \plausible approval profile there exists a committee $W$ that satisfies JR (Theorem 1 in~\cite{ABC+16a}).  Note that $W$ satisfies JR with non-zero probability iff there exists a \plausible  approval profile for which $W$ satisfies JR.

Given a committee $W$, {\sc IsPossJR}
can be solved in polynomial time for the Joint Probability model. The polynomial-time procedure involves taking each \plausible approval profile, which are explicitly provided in the input, in turn and checking whether $W$ satisfies JR for any such profiles. The latter can be done in polynomial time (see Theorem 2 in~\citep{ABC+16a}). This problem is also in complexity class P for the 3VA and candidate-probability models (corollaries of Theorem 13 in~\citep{IBK22a}). Under the lottery model, however,  {\sc IsPossJR} in NP-complete.

\begin{theorem}\label{thm:IsPossJr-Lottery}
{\sc IsPossJR}, the problem of deciding whether a given committee $W$ satisfies JR with non-zero probability,
is NP-complete for the Lottery model, even when each voter 
assigns positive probability to at most three approval sets
and $k=\frac{n}{2}$.
\end{theorem}
\begin{proof}
	{\sc IsPossJR} is in NP, since given a committee $W$ and a \plausible approval profile, we can check in polynomial time whether $W$ satisfies JR (Theorem 2 in~\citep{ABC+16a}).

	To show NP-hardness, we reduce from {\sc 3-SAT}~\citep{GaJo79a}. Let $I$ be a $3CNF$ formula consisting of $n$ clauses and $m$ variables $x_1,\ldots, x_m$. {W.l.o.g. we can assume that $n$ is even.} We reduce to an instance 
	{$I'=(V,C,[\Delta_i],k=\frac{|V|}{2},W)$}
	of {\sc IsPossJR} as follows. We have one voter per clause in $V$, so $|V|=n$. Fix an ordering on the literals in each clause of the $3CNF$. Let $C_1=\{\ell_1, \ldots, \ell_{3n}\}$ where candidate $\ell_1$ corresponds to the first literal in the first clause, candidate $\ell_2$ corresponds to the second literal in the first clause, and so on. Each voter $i$ 
	assigns positive probability to exactly 3 approval sets,
	each corresponding to a literal in the voter's corresponding clause. The $j$'th approval set of voter $i$ includes $\ell_{3(i-1)+j}$. For any two $\ell_i$ and $\ell_j$, $i<j$, that are associated with two distinct clauses, if $\ell_i = \Bar{\ell_j}$ we create a candidate $c_{i,j}$ and add it to the two corresponding approval sets. 
	Let $C_2$ denote the set of all such candidates $c_{i,j}$. Additionally, we create $c_1,\ldots c_k$ dummy candidates and include them all in $W$.
%
{An example illustrating this reduction is provided after the proof.}
We claim that $I$ admits a satisfying assignment iff $W$ satisfies JR for a \plausible approval profile in $I'$. Note that as $W$ has no candidate in common with any of the voters' \plausible approval sets, it satisfies JR iff no $\frac{n}{k}=2$ voters approve of the same candidate. Note that the only candidates that can be approved by more than one voters are {$c_{i,j}$'s}.

	Assume that $I$ admits a satisfying assignment $T$. Therefore at least one literal in each clause is set to True. Consider the approval sets in $I'$ corresponding to True literals. For each voter, there exists at least one approval set corresponding to a True literal. If there are more than one, select one randomly. We claim that in the resulting approval profile, no two voters approve of the same candidate. Assume for contradiction that two voters do approve of the same candidate $c_{i,j}$. This implies that in $T$ a literal and its negation were both set to true, which contradicts that $T$ is a satisfying assignment.

	Now assume that $W$ satisfies JR for a \plausible approval profile $\profile$ in $I'$. For each voter $i$, set the literal corresponding to their approval set to True in $I$. We claim that the resulting truth assignment $T$ is feasible (does not set a literal and its negation to true) and satisfies $I$. Assume for contradiction that $T$ sets a literal and its negation to true. This implies that two voters approve of the same candidate, say $c_{i,j}$, contradicting that $W$ satisfies JR for $\profile$.
	One approval set is picked for each voter, so one literal is set to True in $T$ for each clause, hence $T$ satisfies $I$. The proof is now complete.
\end{proof}

\emph{Example:} To aid the reader in understanding the above proof, we provide the following example illustrating the reduction. Consider the following 3CNF formula: $(x_1\lor x_2 \lor x_3) \land (x_1 \lor x_2 \lor \Bar{x_3}) \land (\Bar{x_3} \lor \Bar{x_2} \lor x_4) \land (x_1 \lor x_2 \lor x_4)$. This instance has 4 clauses and 4 variables. To reduce to an instance of {\sc IsPossJR}, we create a voter per clause, hence $V = \{1,2,3,4\}$. We will have three types of candidates. Let $C_1 = \{\ell_1, \ldots, \ell_{12}\}$, let $C_2 = \{c_{2,8}, c_{3,6}, c_{3,7}, c_{5,8}, c_{8,11} \}$, and add a dummy candidates set $\{c_1, c_2\}$. We have $k=\frac{4}{2} = 2$ and $W=\{c_1,c_2\}$. Each voter assigns a positive probability to exactly 3 approval sets:
\begin{itemize}
 	\item Voter 1:  $\{\ell_1\}$, $\{\ell_2, c_{2,8}\}$, $\{\ell_3, c_{3,6}, c_{3,7}\}$
	\item Voter 2: $\{\ell_4\}$, $\{\ell_5, c_{5,8}\}$, $\{\ell_6, c_{3,6}\}$
	\item Voter 3: $\{\ell_7, c_{3,7}\}$, $\{\ell_8, c_{2,8}, c_{5,8}, c_{8,11}\}$, $\{\ell_9\}$
	\item Voter 4: $\{\ell_{10}\}$, $\{\ell_{11}, c_{8,11}\}$, $\{\ell_{12}\}$
\end{itemize}

\section{Necessary JR}
\label{sec:necessary-JR}
In this section, we examine problems related to the Necessary JR. 
Firstly, we note that {\sc IsNecJR}, the problem of deciding whether a given committee $W$ satisfies JR with probability one, is trivially in P for the joint probability model. The polynomial time procedure involves taking each \plausible approval profile in turn and checking whether $W$ satisfies JR for each \plausible profiles. The latter can be done in polynomial time (see Theorem 2 in~\citep{ABC+16a}). This problem is in P also for the 3VA and candidate-probability models (immediate corollary of Theorem 13 in~\citep{IBK22a}). This problem is in P for the Lottery model too.
\begin{theorem}\label{thm:IsNecJR-Lottery}
 {\sc IsNecJR}, the problem of deciding whether a given committee $W$ satisfies JR with probability 1, is solvable in polynomial time for the Lottery model.
\end{theorem}
\begin{proof}

Consider a committee $W$.
Define the set $V^*\subseteq V$ to be the set of voters who are not represented in $W$ under at least one of their \plausible approval sets. That is, $V^*=\{i: S_r\cap W =\emptyset \text{ for some } r \in [s_i]\}$.
For each $c\in C\setminus W$, we check whether $\frac{n}{k}$ or more voters in $V^*$ have $c$ in at least one of their \plausible approval sets. 
If yes, then there exits a \plausible approval profile for which $W$ does not satisfy JR and hence $W$ does not satisfy JR with probability 1.
If this condition does not hold for any $c\in C\setminus W$, then $W$ clearly satisfies JR under all \plausible approval profiles.
\end{proof}

{\sc ExistsNecJR}, the problem of deciding whether there exists a committee $W$ that satisfies JR with probability 1, is NP-complete for our four models. We prove by reducing from the following problem, which 
we will show is NP-complete.
Note that in the definition of JR (see Definition \ref{def:JR}) the size of $W$ was set to be equal to $k$, the size of the committee to be selected. One can however modify this and let $W$ be of size at most $k$, allowing to choose a committee of size smaller than $k$. 

\begin{definition}\label{def:sizejr}
{\sc SizeJR} is the problem of deciding, given an instance $(V, C, \profile, k)$ of ABC voting and a positive integer $r< k$,  whether there exists a committee $W$, $|W|=r$, such that $W$ satisfies JR.
\end{definition}

\begin{lemma}
\label{lem:sizeJR-NPC}
{\sc SizeJR} is NP-complete.
\end{lemma}
\begin{proof}
{\sc SizeJR} is in NP. Give a committee $W$ one can easily check in polynomial time whether $W$ satisfies JR (Theorem 2 in \citep{ABC+16a}). To show NP-hardness, we reduce from the NP-complete problem of {\sc CandidateCover} which is to decide, given an instance $I=(V,C,\profile,r)$ where $V$ is a set of voters, $C$ a set of candidates, $\profile$ an approval profile where $A_i\neq \emptyset, \forall i\in V$, and $r\leq |C|$ a positive integer, whether there exists a committee $W$ of size $r$ such that each voter is represented in $W$ , i.e. $A_i \cap W \neq \emptyset$. (Note that {\sc CandidateCover} is {\sc SetCover}~\citep{GaJo79a} described in voting terminology.)

We reduce from an instance $I$ 
of {\sc CandidateCover} to an instance $I'=(V, C, \profile, k=|V|, r)$ of SizeJR.
To satisfy JR every subset $V^*$ of voters, $|V^*| \geq \frac{|V|}{k}=1$, has to be represented, implying that every voter has to be represented. Thus it is straightforward to see that every YES instance of $I$ is a YES instance of $I'$ and vice versa.
\end{proof}

\begin{lemma}\label{lem:fullset}
Let {$I=(V, C, [p_{i,c}],k)$} be an instance of ABC voting under Candidate-Probability model. If $0 < p_{i,c} <1$ for all voters $i$ and candidates $c$, then a committee $W$ 
satisfies JR with probability 1
iff $W=C$ and $k=|C|$.
\end{lemma}
\begin{proof}
Take any candidate $c\in C$. Let $\profile_c$ be the approval profile in which each voter's approval set is exactly $\{c\}$. $\profile_c$ is a \plausible approval profile. As all the voters only approve of $c$, to satisfy JR for $\profile_c$, $W$ must contain $c$. Our argument was independent of the choice of $c$, hence for $W$ to satisfy JR for all \plausible approval profiles it must be that $W=C$ and $k=|C|$.
\end{proof}
\begin{theorem}\label{thm:ExistsNecJRCandidate}
{\sc ExistsNecJR}, the problem of deciding whether there exists a committee $W$ that satisfies JR with probability 1, 
is NP-complete for the Candidate-Probability model.
\end{theorem}
\begin{proof}
{\sc ExistsNecJR} is in NP because {\sc IsNecJR} is in P (immediate corollary of Theorem 13 in~\citep{IBK22a}).
To prove NP-hardness we reduce from {\sc SizeJR} that we have shown is NP-complete (see \Cref{lem:sizeJR-NPC}).

Given an instance $I=(V, C, \profile, k, r)$ of {\sc SizeJR},  we reduce to an instance {$I'=(V', C', [p_{i,c}], k')$} of {\sc ExistsNecJR} under the Candidate-Probability model as follows. 
Suppose $|V|=n$ and $|C|=m$.
 Let $V' = V \cup V^+$ where in $V^+$ we have $n$ new voters. Let $C'=C\cup C^+$ where in $C^+$ we have $2k-r$ new candidates 
 {$\{1+m,\ldots, 2k-r+m\}$}.
 Each voter in $V$ has the same approval set as in $\profile$, implying that for each voter $i\in V$, $p_{i,c} = 1$, $\forall c\in A_i$ and $p_{i,c'}=0$, $\forall c'\in C'\setminus A_i$. Each voter $i\in V^+$ disapproves of all candidates in $C$ (i.e. $p_{i,c}=0, \forall c\in C$) and finds all new candidates possibly acceptable (i.e. $0< p_{i,c}<1, \forall c\in [1+m,2k-r+m]$). Let $k'=2k$. As $|V'| = 2|V|$ we have that $\frac{|V'|}{k'} = \frac{|V|}{k}=\frac{n}{k}$. This reduction can be done in polynomial time. It remains to show that every YES instance of $I$ is a YES instance of $I'$ and vice versa.

Suppose $I$ is a YES instance of {\sc SizeJR}. Then there exists a committee $W$, $|W|=r$, that satisfies JR in $I$.
Let $W'=W \cup C^+$. We claim that $W'$ satisfies JR under all \plausible approval profiles of $I'$. Note that $|W'| = r + 2k-r = 2k$ and hence of the correct size. 
By construction we have that $C^+ \subset W'$, therefore all voters in $V^+$ are represented in $W'$. 
Voters in $V$ 
do not approve of any candidate in $C^+$. 
$W$ satisfies JR in $I$, implying that there does not exist a candidate $c\in C\setminus W$ who is approved by at least $\frac{n}{k}$ voters in $V$ none of whom is represented in $W$. Therefore $W'$ satisfies JR under all 
\plausible approval profiles
and hence $I'$ is a YES instance.

Suppose that $I'$ is a YES instance of {\sc ExistsNecJR}. Then there exists a committee $W'$, $|W'|=2k$, that satisfies JR under all \plausible approval profiles of $I'$. Voters in $V^+$ do not approve of any candidates in $C$, hence it follows from Lemma \ref{lem:fullset} that $C^+ \subset W'$. Let $W = W' \setminus C^+$ and so we have that $|W| = 2k-(2k-r) = r$. Voters in $V$ do not approve of any candidates in $C^+$, so they can only be represented by candidates in $C$. Therefore, for $W'$ to satisfy JR under all \plausible approval profiles it must be that there does not exist a candidate $c\in C\setminus W$ who is approved by at least $\frac{n}{k}$ voters in $V$ none of whom is represented in $W$. Hence $W$ satisfies JR in $I$ and $I$ is a YES instance.
\end{proof}
\begin{corollary}\label{thm:ExistsNecJR3VA}
{\sc ExistsNecJR} is NP-complete for the 3VA models.
\end{corollary}


We can show that {\sc ExistsNecJR} is NP-complete for the Lottery model and Joint Probability model, using an almost identical argument as to the one presented in the proof of Theorem~\ref{thm:ExistsNecJRCandidate}. The main variation is the instance reduction due to a different problem setting. 

\begin{theorem}\label{thm:ExistsNecJRLotteryJoint}
{\sc ExistsNecJR}, the problem of deciding whether there exists a committee $W$ that satisfies JR with probability 1, is NP-complete for the Lottery model and Joint Probability model.
\end{theorem}
\begin{proof}
We present the instance reductions in the settings of Lottery and Joint Probability models below, building upon the proof argument mentioned earlier in Theorem~\ref{thm:ExistsNecJRCandidate}.

\noindent\emph{\textbf{ Lottery model:}} Given an instance $I=(V, C, \profile, k, r)$ of {\sc SizeJR}, we reduce to an instance {$I'=(V', C', [\Delta_i], k')$} of {\sc ExistsNecJR} under the Lottery model as follows. Suppose $|V|=n$. Let $V' = V \cup V^+$ where in $V^+$ we have $n$ new voters. Let $C'=C \cup C^+$ where in $C^+$ we have $2k-r$ new candidates  {$\{1+m,\ldots, 2k-r+m\}$}. Each voter in $V$ approves of $A_i$ with probability 1, i.e. {$\Delta_i(A_i) = 1$}. For each voter $i\in V^+$ we have {$\Delta_i(\{j+m\})=\frac{1}{2k-r}$}, for all $j\in [1,2k-r]$. Let $k'=2k$. 

\noindent\emph{\textbf{Joint Probability model:}} Given an instance $I$
of {\sc SizeJR},  we reduce to an instance {$I'=(V', C', \Delta, k')$} of {\sc ExistsNecJR} under the Joint Probability model as follows. 
$V'$ and $C'$ are defined as in above.
Let {$\Delta(A_1, \ldots, A_n, \{j+m\}, \dots ,\{j+m\}) = \frac{1}{2k-r}$}, for all $j\in [1,2k-r]$. That is, there are exactly $2k-r$ \plausible approval profiles, each voter $i\in V$ approves of $A_i$ under all \plausible approval profiles, and the new voters all approve of exactly one new candidate, and only that candidate, in exactly one \plausible approval profile. Let $k'=2k$. 

The above reductions can be done in polynomial time. It is easy to see, following similar arguments presented in the proof of Theorem~\ref{thm:ExistsNecJRCandidate}, that every YES instance of $I$ is a YES instance of $I'$ and vice versa.
\end{proof}

We observe that, under the Joint Probability model, {\sc ExistsNecJR} is NP-complete even when minimal uncertainty is present.  

\begin{theorem}\label{thm:ExistsNecJRjoint2}
{\sc ExistsNecJR} is NP-complete for the Joint Probability model even if there are only two approval profiles $\profile_1$ and $\profile_2$ associated with a positive probability; i.e., $s=2$ and $\lambda_1=1-\lambda_2$.
\end{theorem}
\begin{proof}
Given an instance $I=(V, C, \profile, k, r)$ 
of {\sc SizeJR},  we reduce to an instance $I'=(V', C', \Delta, k')$ of {\sc ExistsNecJR} under the Joint Probability model as follows. Suppose $|V|=n$. Let $V' = V \cup V^+$ where in $V^+$ we have $n$ new voters and $C'=C\cup C^+$ where in $C^+$ we have $2k-r$ new candidates.  
Let $\Delta(\profileset) = \{(\lambda_1, \profile_1\}, (\lambda_2, \profile_2\})$, such that $0<\lambda_1<1$ and $\lambda_2=1-\lambda_1$.That is, there are exactly 2 \plausible  
approval profiles. Let
$\profile_1 = (A_1, \ldots, A_n, \frac{n}{k}\times\{m+1\}, \frac{n}{k}\times\{m+2\} \dots,\frac{n}{k}\times \{m+\frac{2k-r}{2}\}, \{\}, \ldots \{\})$ and $\profile_2 = (A_1, \ldots, A_n, \frac{n}{k}\times\{m+\frac{2k-r}{2}+1\}, \dots, \frac{n}{k}\times\{m+(2k-r)\}, \{\}, \ldots, \{\})$. The notation $\frac{n}{k}\times \{c\}$ indicates that exactly $\frac{n}{k}$ voters approved a candidate set $\{c\}$.  Note that empty approval sets represent the dummy voters who do not approve of any candidates and are added to make the total number of voters $2n$. Each voter $i\in V$ approves of $A_i$ under both approval profiles. All the new voters approve of exactly one new candidate, and every new candidate is approved by $\frac{n}{k}$ new voters.  Let $k'=2k$. As $|V'| = 2n$ we have that $\frac{|V'|}{k'} =\frac{n}{k}$. This reduction can be done in polynomial time. We claim that every YES instance of $I$ is a YES instance of $I'$ and vice versa.

Suppose $I$ is a YES instance of {\sc SizeJR}. Then there exists a committee $W$, $|W|=r$, that satisfies JR in $I$. Let $W'=W \cup  C^+$. We claim that $W'$ satisfies JR under all \plausible approval profiles of $I'$. Note that $|W'| = r + 2k-r = 2k$ and hence of the correct size. As $C^+ \subset W'$, it follows from the construction of the approval profiles that all voters in $V^+$ are represented in $W'$. $W$ satisfies JR in $I$, implying that there does not exist a candidate $c\notin W$ who is approved by at least $\frac{n}{k}$ voters in $V$ none of whom is represented in $W$. Therefore $W'$ is necessarily JR in $I'$ and $I'$ is a YES instance.

Suppose that $I'$ is YES instance of {\sc ExistsNecJR}. Then there exists a committee $W'$, $|W'|=2k$, that satisfies JR under all \plausible approval profiles of $I'$. It follows from the construction of the \plausible approval profiles that for any given $c\in C^+$, it is possible that $\frac{n}{k}$ voters in $V^+$ have $\{c\}$ as their approval set. Therefore, for $W'$ to satisfy JR under all \plausible approval profiles, it must be that $C^+ \subset W'$. Let $W = W' \setminus C^+$ and so we have that $|W| = 2k-(2k-r) = r$. Voters in $V$ do not approve of any candidates in $C^+$, so they can only be represented by candidates in $C$. Therefore, for $W'$ to satisfy JR under all \plausible approval profile, it must be that there does not exist a candidate $c\notin W$ who is approved by at least $\frac{n}{k}$ voters in $V$ none of whom is represented in $W$. Hence $W$ satisfies JR in $I$ and $I$ is a YES instance.
\end{proof}

{\sc ExistNecJR} is in P for the Lottery model when each voter approves of only one candidate in each \plausible approval profile. 
\begin{theorem}\label{thm:IsNecJR-Lottery-P}
 {\sc ExistNecJR} is in P for the Lottery model when each voter specifies a probability distribution over approval sets, with each set containing a single candidate.
\end{theorem}
\begin{proof}
For each candidate, we count the number of voters who potentially approve it. For $W$ to satisfy JR under all \plausible approval profiles, any candidate $c$ with a count $\ge \frac{n}{k}$ must be in $W$. This is because if $c$ is not in $W$ then $W$ does not satisfy JR under the \plausible approval profile in which all voters who potentially approve of $c$ have $\{c\}$ as their approval set. Therefore, if there are more than $k$ candidates with a count $\ge \frac{n}{k}$, then there is no $W$ that satisfies JR with probability 1. Otherwise it is easy to see that any winning committee $W$ that contains all candidates with a count $\ge \frac{n}{k}$ satisfies JR with probability 1. 
\end{proof}

\section{ {\sc JR-Probability}}
\label{sec:JR-probability}

{\sc JR-Probability} is the problem of determining the probability that a given $W$ satisfies JR, which is defined as follows:
$p (W~ \text{satisfies}~ JR) = \sum_{\profile\in \profileset} \Delta(\profile) \times \mathcal{I} ( \text{\sc IsJR}(W, \profile) ), 
$
where $\mathcal{I}$ is an indicator function that returns 1 when {\sc IsJR}$(W, \profile)$ is a YES instance, i.e. $W$ satisfies JR for $\profile$, and 0 otherwise.

Given a set $W$ and an approval profile $\profile$, we can check in polynomial time whether $W$ satisfies JR or not (see Theorem 2 in~\citep{ABC+16a}). The Joint Probability model takes approval profiles and their corresponding probability distributions as input. Therefore,
\begin{theorem}\label{thm:JR-prob-CP}
{\sc JR-Probability } is in P for the Joint probability model.
\end{theorem}

The decision variant of the {\sc JR-Probability} involves assessing whether the JR probability of $W$ exceeds a designated threshold, denoted as $\tau$. Notably, when $\tau$ is set to zero, the problem transforms into {\sc IsPossJR} which is NP-complete (Theorem~\ref{thm:IsPossJr-Lottery}). Therefore,

\begin{corollary}\label{cor:JR-prob-lottery}
{\sc JR-Probability}  is NP-hard for the Lottery model.
\end{corollary}

{{\sc JR-Probability} is \#P-complete for the 3VA model, and hence consequently for the Candidate Probability model. We prove by reducing from \#VertexCovers, the problem of counting the number of vertex covers for a given graph, that we will show is \#P-complete by a short proof.}

\begin{lemma}
\#VertexCovers (\#VC) is \#P-complete. 
\end{lemma}
\begin{proof}
For any given graph, a set $S\subseteq V$ is a vertex cover if and only if the subset $V\setminus S$ is an independent set~\citep{GaJo79a}. It is well known that counting independent sets of a graph is  \#P-complete~\citep{Green00a}.
\end{proof}

{Solving {\sc JR-Probability} for the 3VA model involves counting the number of \plausible approval profiles for which the given $W$ satisfies JR.} 

%
\begin{theorem}\label{th:counting}
  {\sc JR-Probability} is \#P-complete for the 3VA model.

\end{theorem}
\begin{proof}
We prove the statement by proving that    {\sc \#IsJR}, the problem of counting the number of \plausible approval profiles for which a given committee $W$ satisfies JR, is \#P-complete for the 3VA model.  Since the total number of  \plausible approval profiles can be computed  easily (the count is $\prod_{i\in V}2^{x_i}$ where $x_i$ is $|\{c\in C\mid p_{i,c}=0.5\}|$) and all \plausible approval profiles are equiprobable, the theorem will follow. 

Consider an instance $I' = G(V', E')$ of the {\sc \#VC }, where a graph $G$ is given with $V'$ as the set of $n$ vertices and $E'$ as the set of $m$ edges. 
We reduce to an instance  $ I = (V, C, [p_{i,c}], k, W)$ of {\sc \#IsJR} such that there is a one to one correspondence between 
vertex covers and \plausible approval profiles for which $W$ satisfies JR. We have one voter per vertex, so $|V|=n$.  For each edge $(i,j)$, $i<j$, in $E'$ we create a candidate $c_{i,j}$ and have voters $i$ and $j$ approve of $c_{i,j}$ with probability 1. Other voters disapprove of $c_{i,j}$. That is, $p_{i,c_{i,j}} = p_{j,c_{i,j}} =1$ and $p_{t,c_{i,j}} = 0, \forall t\in V\setminus \{i,j\}$.
%
We additionally create a set of $k = \frac{n}{2}$ candidates $C^+ = \{ c^+_1, c^+_2, ..., c^+_k \}$, so we have $|C|=m+k$. Each voter approves of $c^+_1$ with probability 0.5, i.e., $p_{i, c^+_1} = 0.5, \forall i\in V$, and disapproves of the rest of the candidates in $C^+$, i.e., $p_{i, c^+_j} = 0, \forall i\in V, \forall j\in[2,k]$. We set $W = C^+$.  Note that $\frac{n}{k} = 2$. This reduction can be done in polynomial time.

We claim that every vertex cover in $I'$ corresponds to a unique \plausible approval profile in $I$ for which $W$ satisfies JR, and vice versa.
Given a subset of vertices $S\subseteq V'$, we define a \plausible approval profile $\profile^S$ as follows. Voters $V^S\subseteq V$ corresponding to $S$ approve of $c^+_1$ 
and the remaining voters disapprove of $c^+_1$.

We first show that if $S$ is a vertex cover then $W$ satisfies JR with respect to $\profile^S$. $S$ being a vertex cover implies that at least one end point of every edge is in $S$. This implies that, in $\profile^S$, of any two voters who approve of the same candidate, at least one approves of $c^+_1$ also, and is hence represented in $W$. Hence $W$ satisfies JR.
%
Conversely, we show that if $S$ is not a vertex cover then $W$ does not satisfy JR in $\profile^S$. $S$ not being a vertex cover implies that there is at least one edge $(i,j)$ whose end points $i$ and $j$ are not in $S$. This implies that, in $\profile^S$, the corresponding voters $i$ and $j$ both approve of the same candidate $c_{i,j}$ and disapprove of $c^+_1$. Hence there is a set of candidates of size $\frac{n}{k} = 2$ who approve of the same candidate and neither is represented in $W$, implying that $W$ does not satisfy JR in $\profile^S$.
\end{proof}

 \begin{theorem}\label{thm:JR-prob-3VA-Pcase1}
{\sc JR-Probability} is in P for the 3VA model if $p_{i,c}\in \{0,1\}$, $\forall i \in V$ and $\forall c\in W$.
\end{theorem}
\begin{proof}
Let $V'$ be the set of voters who have not approved any candidate in $W$. Define $n^+_j$ (respectively, $n^1_j$ and $n^u_j$) as the number of voters from $V'$ who approved candidate $j$ with $p_{i,j} > 0$ (respectively, $p_{i,j} = 1$ and $p_{i,j} = 0.5 \}$). So $n^+_j = n^1_j + n^u_j$.
 If $\exists j \in C\setminus W$ such that $n^1_j \ge \frac{n}{k}$, then the probability of $W$ being JR is zero.
 In the case where $n^1_j < \frac{n}{k}$, we utilize the following procedure to calculate the probability of $W$ being JR. Note that since all \plausible approval profiles are equiprobable under 3VA, it is sufficient to count the number of \plausible approval profiles and those that violate JR.

For each candidate $j \in C\setminus W$, we count the number of plausible assignments of voters' unknown approvals under which $W$ does not satisfy JR because of $j$.
A plausible assignment is obtained by setting $p_{i, j}$  to 0 or 1 wherever $p_{i, j} = 0.5$. This results in $2^{n^u_j}$ plausible assignments. Out of these, we count the number of assignments that violate JR because of $j$, denoted as $t_j$, as follows:
$t_j = \sum_{l= \frac{n}{k} - n^1_j}^{n^u_j} {n_j^u \choose l}.$
Therefore, the probability that $W$ violates JR because of candidate $j$ is $p_j(W \text{ violates } JR) = \frac{t_j}{2^{n^u_j}}$. When $n^+_j < \frac{n}{k}$, $t_j = 0$.  
The probability of $W$ satisfying JR can be expressed as the product of individual probabilities, where each term corresponds to 1 minus the probability of $W$ violating JR because of each candidate $j \in C\setminus W$:
$p(W \text{ being JR })  = \prod_{j\in C\setminus W} \big(1 - p_j(W \text{ violates } JR)\big).$
 This computation can be done in polynomial time.
\end{proof}

The above process ultimately computes the count of \plausible approval profiles where $W$ satisfies JR, divided by the total number of \plausible approval profiles. 
It is essential to note that the same computation will not work when each profile has a different probability, as is the case in the candidate probability model. 
We also observe that for the 3VA model, the computation of {\sc JR-Probability} is achievable in polynomial time if $n=k$.
\begin{theorem}\label{thm:JR-prob-3VA-Pcase2}
In the 3VA model, counting the \plausible approval profiles where $W$ satisfies JR can be computed in polynomial time if $k = n$.
Hence, {\sc JRProbability} is in P when  $k=n$.
\end{theorem}
\begin{proof}
When $k=n$, we have $\frac{n}{k} = 1$ and so ensuring JR for a \plausible approval profile requires that each voter who approves at least one candidate must have at least one of their approved candidates represented in the winning committee $W$. Hence, for each voter $i$, the objective is 
to count the number of \plausible approval sets for $i$ under which either $i$ is represented in $W$ or $i$ does not approve of any candidate.
We let $t_i$ denote this count which can be computed in polynomial time as follows. If $i$ approves of any candidate in $W$ with probability 1 then $t_i$ will be $2^{x_i}$ where $x_i$ is the number of candidates $i$ is unsure about. Otherwise (if $i$ doesn't assign probability 1 to any candidate in $W$) then let $y_i$ be the number of candidates in $W$ that $i$ is unsure about. Then $t_i = (2^{y_i} -1) \times 2^{x_i-y_i}$ if there exists a candidate that $i$ approves with probability 1, and  $t_i = (2^{y_i} -1) \times 2^{x_i-y_i} +1$ otherwise (the last term counts for the case where $i$ does not approve of any candidate.) 
The overall number of \plausible approval profiles where $W$ satisfies JR is the product of these individual counts, i.e., $\prod_i {t_i}$.

%
\end{proof}

Note that {\sc ExistsNecJR} reduces to {\sc MaxJR}. The following result then follows from the NP-completeness of  {\sc ExistsNecJR} for all of our four uncertainty models. 
\begin{corollary} \label{cor:MaxJR}
{\sc MaxJR} is NP-hard for all of our four uncertainty models.
\end{corollary}

\section{Results for PJR and EJR}
\label{sec:PJR-and-EJR}

In this section, we provide answers to the computational problems under the stronger concepts of PJR and EJR. Our proofs rely on the coNP-complete problems of {\sc IsEJR} and {\sc IsPJR} (see ~\citep{ABC+16a,AEH+18a}). {\sc IsEJR} (resp. {\sc IsPJR}) is the problem of deciding, given an instance $(V,C,\profile,k)$ of ABC voting and a committee $W$, whether $W$ satisfies EJR (resp. PJR). The coNP-completeness of these problems gives us the following result immediately.

\begin{corollary}\label{cor:IsNecEJR-IsPossEJR}
{\sc IsNecEJR} (resp. {\sc IsNecPJR}), the problem of deciding whether a given committee $W$ satisfies EJR (resp. PJR) with probability 1, and {\sc IsPossEJR} (resp. {\sc IsPossPJR}), the problem of deciding whether a given committee $W$ satisfies EJR (resp. PJR) with non-zero probability, are both coNP-complete for all of our four uncertainty models.
\end{corollary}


\begin{theorem}\label{thm:ExistsNecEJR-3VA-CP}
{\sc ExistsNecEJR} (resp. {\sc ExistsNecPJR}), the problem of deciding whether there exists a committee $W$ that satisfies EJR (resp. PJR) with probability 1 is coNP-hard for the Candidate-Probability and 3VA models.
\end{theorem}

\begin{proof}
We provide the proof for {\sc ExistsNecEJR} . The proof for {\sc ExistsNecPJR} is identical and requires only to substitute {\sc IsEJR} with {\sc IsPJR}.
We prove {\sc ExistsNecEJR} is coNP-hard by reducing from {\sc IsEJR} that is known to be coNP-complete ~\citep{ABC+16a}.

Given an instance $I=(V, C, \profile, k, W)$ of {\sc IsEJR},  we reduce to an instance $I'=(V', C', [p_{i,c}], k')$ of {\sc ExistsNecEJR} under the Candidate-Probability model as follows. Suppose $|V|=n$. Let $V' = V \cup V^+$ where in $V^+$ we have $n$ new voters. Let $C'=C\cup C^+$ where in $C^+$ we have $k$ new candidates $\{c^+_{1},\ldots, c^+_{k}\}$. Each voter in $V$ has the same approval set as in $\profile$, implying that for each voter $i\in V$, $p_{i,c} = 1$, $\forall c\in A_i$ and $p_{i,c'}=0$, $\forall c'\in C'\setminus A_i$. Each voter $i\in V^+$ disapproves of all candidates in $C\setminus W$ (i.e. $p_{i,c}=0, \forall c\in C\setminus W$) and finds all the other candidates possibly acceptable (i.e. $0< p_{i,c}<1, \forall c\in W\cup C^+$)\footnote{Set $p_{i,c} = 0.5$, in case of 3VA model.}. Let $k'=2k$. As $|V'| = 2|V|$ we have that $\frac{|V'|}{k'} = \frac{n}{k}$. This reduction can be done in polynomial time. It remains to demonstrate that every YES instance of $I$ is a YES instance of $I'$, and vice versa.

Suppose $I$ is a YES instance of {\sc IsEJR}. Then $W$ satisfies EJR in $I$. Let $W'= W \cup C^+$. Clearly $|W'|=2k=k'$. We claim that $W'$ satisfies EJR under all \plausible approval profiles of $I'$. Voters in $V^+$ do not approve of any candidate in $C\setminus W$; by construction, we have that $C^+ \subset W'$ and $W\subset W'$. Therefore, 
$W'$ satisfies EJR for every $\ell\in[1,k]$ and every $\ell$-cohesive group of voters in $V'$.
Therefore, $W'$ satisfies EJR with probability 1 and $I'$ is a YES instance.

If $I$ is a NO instance of {\sc IsEJR}, then $W$ violates EJR w.r.t. a set of voters in $V$. 
In $I'$, voters in $V^+$ do not approve of any candidate in $C\setminus W$ and can possibly approve every candidate in $C^+\cup W$. Consequently, to satisfy EJR under all \plausible approval profiles, a winning committee $W'$ must include all candidates in $C^+\cup W$.  
As $|C^+\cup W| = 2k = k'$, $W'$ cannot include any other candidate. Since voters from $V$ do not approve of any candidates from $C^+$ and $W$ does not satisfy EJR with respect to $V$ and $\frac{n}{k}$, 
$W'$ cannot satisfy JR for {any \plausible approval profile}.
As a result, $I'$ is a NO instance of {\sc ExistsNecEJR}.
\end{proof}

We can show that {\sc ExistsNecEJR} (resp. {\sc ExistsNecPJR}) is coNP-hard for the Lottery and Joint Probability models, using an almost identical argument as to the one presented in the proof of Theorem~\ref{thm:ExistsNecEJR-3VA-CP} . The main variation is the instance reduction due to a different problem setting. We omit the proof due to shortage of space. 

\begin{theorem}\label{thm:ExistsNecEJR-Lottery-JP}
{\sc ExistsNecEJR} (resp. {\sc ExistsNecPJR})
is coNP-hard for the Lottery and Joint Probability models.
\end{theorem}
\begin{proof}
The proof argument is identical to that presented in Theorem~\ref{thm:ExistsNecEJR-3VA-CP}, with the only variation being the instance reduction due to a different problem setting. Therefore, we present the instance reductions in the settings of lottery and joint probability models below, building upon the proof argument mentioned earlier in Theorem~\ref{thm:ExistsNecEJR-3VA-CP}.

\noindent\emph{\textbf{ Lottery model:}} Given an instance $I=(V, C, \profile, W, k)$ of {\sc IsEJR},  we reduce to an instance  $I'=(V', C', [\Delta_i], k')$ of {\sc ExistsNecEJR} under the Lottery model as follows. Suppose $|V|=n$. Let $V' = V \cup V^+$ where in $V^+$ we have $n$ new voters. Let $C'=C\cup C^+$ where in $C^+$ we have $k$ new candidates $\{c^+_{1},\ldots, c^+_{k}\}$. Each voter in $V$ has the same approval set as in $\profile$, implying that for each voter $i\in V$, $\Delta_i(A_i) = 1$. For each voter $i\in V^+$ we have $\Delta_i(\{c_j\}) = \frac{1}{2k},~ \forall c_j\in W\cup C^+$. Let $k'=2k$. As $|V'| = 2|V|$ we have that $|V'|/k' = |V|/k=n/k$. 

\noindent\emph{\textbf{Joint Probability model:}} Given an instance $I=(V, C, \profile, W, k)$ of {\sc IsEJR},  we reduce to an instance  $I'=(V', C', \Delta, k')$ of {\sc ExistsNecEJR} under the Joint Probability model as follows. Suppose $|V|=n$. Let $V' = V \cup V^+$ where in $V^+$ we have $n$ new voters. Let $C'=C\cup C^+$ where in $C^+$ we have $k$ new candidates $\{c^+_{1},\ldots, c^+_{k}\}$. Let $\Delta(A_1, A_2, \ldots, A_n, \{c_j\}, \ldots, \{c_j\}) = \frac{1}{2k}$, $\forall c_j\in W\cup C^+$. Let $k'=2k$. As $|V'| = 2|V|$ we have that $|V'|/k' = |V|/k=n/k$.

The above reductions are polynomial time. It is easy to see, following similar arguments presented in the proof of Theorem~\ref{thm:ExistsNecEJR-3VA-CP}, that every YES instance of I is a YES instance of I' and vice versa under both models.
\end{proof}

{\sc MaxEJR} (resp. {\sc MaxPJR}) is the problem of computing a committee $W$ that has the highest probability of satisfying EJR (resp. PJR). Note that {\sc ExistsNecEJR} (resp. {\sc ExistsNecPJR}) reduces to {\sc MaxEJR} (resp. {\sc MaxPJR}). 
The following result then follows from the coNP-hardness of {\sc ExistsNecEJR} (resp. {\sc ExistsNecPJR}).
\begin{corollary} \label{cor:MaxEJR-MaxPJR}
{\sc MaxEJR} (resp. {\sc MaxPJR}) is coNP-hard for all of our four uncertainty models.
\end{corollary}
\begin{proof}
{\sc ExistsNecEJR} (resp. {\sc ExistsNecPJR}), the problem of finding a committee $W$ with a probability of being EJR (resp. PJR) equal to 1 was proven to be coNP-hard for all the four uncertainty models: 3VA, Candidate Probability, Lottery and Joint Probability models (refer Table~\ref{table:summary:uncertainABC}). If there exist a polynomial time algorithm to solve  {\sc MaxEJR} (resp. {\sc MaxPJR}) in any of these four models, it can be used to solve {\sc ExistsNecEJR} (resp. {\sc ExistsNecPJR}) in the corresponding model. Thus, {\sc MaxEJR} and {\sc MaxPJR} are coNP-hard for the 3VA, Candidate Probability, Lottery and Joint Probability models. 
\end{proof}
\vspace{-1em}
\section{Conclusion}

We have formalized uncertain preferences models for a well-studied committee selection problem. We focused primarily on a fundamental concept of proportional representation called justified representation (JR). Our work suggests several directions including exploring other axiomatic properties such as Pareto optimality, welfare maximization, or other notions of fairness.  In our study, the central problem is {\sc MaxJR}. It will be interesting to explore approximation or fixed-parameter tractable (FPT) algorithms for this problem. Finally, another direction is to explore restricted versions of approval preferences~\citep{ELP22a,ZKO21a}.

\paragraph{Acknowledgments}
This work was supported by the NSF-CSIRO grant on “Fair Sequential Collective Decision-Making" (Grant  No. RG230833). Venkateswara Rao Kagita is grateful to Science and Engineering Research Board (SERB) for providing financial support for this research through the SERB-SIRE project (Project Number: SIR/2022/001217).  Mashbat Suzuki is supported by the ARC Laureate Project FL200100204 on ``Trustworthy AI''. Baharak Rastegari is grateful to UNSW Sydney's Faculty of Engineering for supporting her research visit through Diversity in Engineering Academic Visitor Funding Scheme.

\vspace{-1em}


\bibliographystyle{splncs04nat}
\bibliography{haris_master_ABC,aziz_personal_ABC,ijcai24}



%
%
%
\end{document}